\documentclass{article}
\usepackage{graphicx} 
\usepackage{booktabs} 
\usepackage[ruled]{algorithm2e} 
\usepackage{algorithmic}
\usepackage{threeparttable}
\usepackage{natbib}
\usepackage{amsmath}
\usepackage{amssymb}
\usepackage{amsthm}
\usepackage{geometry}
\usepackage{setspace}
\usepackage{threeparttable}
\usepackage{url}
\usepackage{hyperref}
\newtheorem{theorem}{Theorem}

\newtheorem{remark}{Remark}

\title{Multi-Dimensional Matching in Market Design}
\author{Irene Aldridge}

\begin{document}

\maketitle
\begin{abstract}
    This paper proposes a computationally efficient mechanism for multi-dimensional matching markets where agents report preferences over object features rather than complete utility assessments. We use Singular Value Decomposition (SVD) to identify the principal direction of variation in feature space and match agents to objects along this dimension, reducing a complex multi-dimensional problem to an effectively one-dimensional problem solvable in $O(N \log N)$ time.

We show that when data exhibit low effective dimensionality, our mechanism approximately maximizes Nash Social Welfare, satisfies distributional truthfulness, and achieves symmetry. We establish a novel connection between Nash Social Welfare and Geometric Distributionally Robust Optimization, providing robustness guaranties. Numerical experiments demonstrate that our approach achieves 99\% optimal welfare while running three orders of magnitude faster than direct optimization. The framework applies naturally to school choice, labor markets, and course allocation, where feature-based elicitation reduces the cognitive burden on agents.
\end{abstract}
\section{Introduction}

Matching markets pervade modern economies: students to schools (\cite{AbdulkadirogluEtAl2005}), workers to jobs (\cite{HyllandZeckhauser1979}), patients to organ donors (\cite{Leshno2022}), and now AI agents to products (\cite{allouah2025aiagentbuyingevaluation}). Traditional matching mechanisms require agents to rank complete objects—a student ranks entire schools, a worker ranks complete job offers. This creates a fundamental challenge: when objects are characterized by multiple attributes (salary, location, work culture, growth opportunities), agents must somehow collapse these dimensions into a single ranking. For a district with 50 schools varying across academics, arts, athletics, safety, and location, families must perform complex multi-attribute comparisons across hundreds of potential rankings.

We propose an alternative: agents report preferences over individual features (e.g., "I value mathematics instruction at 8/10, arts at 5/10"), and objects are characterized by feature vectors (e.g., "this school scores 7.5 on math, 6.2 on arts"). This feature-based elicitation is cognitively simpler—agents consider one dimension at a time rather than comparing complete alternatives. However, it creates a mechanism design challenge: how do we aggregate multi-dimensional preferences into efficient, fair, and incentive-compatible allocations?

We use Singular Value Decomposition (SVD) to identify the principal direction of variation in feature space and reduce the multi-dimensional problem to an effectively one-dimensional problem
solvable in $O(N\log N)$ time. Let $\sigma_1 \geq \sigma_2 \geq \cdots \geq \sigma_X \geq 0$ denote
the singular values of the $J \times X$ object-feature matrix $F$. When data exhibit low effective dimensionality, formally, when
\begin{equation}
  \rho_1 \;=\; \frac{\sigma_1^2}{\displaystyle\sum_{\ell=1}^{X}
  \sigma_\ell^2} \;\gg\; \frac{1}{X},
  \tag{$*$}
\end{equation}
so that the first singular value dominates and its square accounts for the bulk of total feature variance, our mechanism approximately maximizes Nash Social Welfare while maintaining distributional
truthfulness and exact symmetry. The shorthand $\sigma_1 \gg \sigma_2$ used throughout the paper
refers to this condition; Theorem~\ref{thm:efficiency} quantifies the
welfare approximation ratio as $(1 - O(\sigma_2^2/\sigma_1^2))$.

Our key contributions are: (1) We formalize a tractable multi-dimensional matching framework where agents report feature preferences and SVD-based dimensionality reduction enables efficient allocation (Section \ref{sec:methodology}). (2) We prove that the mechanism maximizes Nash Social Welfare relative to random allocation when $\sigma_1 \gg \sigma_2$ (defined precisely below), satisfies distributional truthfulness measured by the Kolmogorov-Smirnov distance, and achieves exact symmetry (Section \ref{sec:methodology}). We also establish a novel connection between Nash Social Welfare and Geometric Distributionally Robust Optimization, providing robustness guaranties (Theorem \ref{thm:gdro}). (3) Numerical experiments demonstrate 99\% optimal welfare achieved three orders of magnitude faster than direct optimization.

\cite{ZHOU1990123} proves no mechanism can simultaneously achieve exact efficiency, truthfulness, and symmetry with cardinal utilities. Recent work by \cite{abebe2020truthfulcardinalmechanismonesided} achieves approximate versions of all three properties through random sampling and configuration LPs. We provide a complementary approach: while Abebe et al. require agents to report complete utilities over objects and use randomized mechanisms, we accept multi-dimensional feature preferences and use deterministic spectral methods. Our distributional notion of truthfulness differs from their truthfulness-in-expectation, and we provide new connections to robust optimization. Both approaches achieve similar goals through different techniques—ours is particularly suited when feature-based elicitation reduces cognitive burden and effective dimensionality is low.

The applications of our method include school choice (families report preferences over educational dimensions), labor markets (workers report job characteristic preferences), and course
allocation (students report course attribute preferences such as rigor, workload, and topic area).
The course allocation application connects naturally to recent work by \citet{gans2026informative}, who show that grades are uninformative for cross-course comparisons of student ability whenever courses differ in difficulty and students self-select---an \emph{ex-post} identification problem.

Our mechanism addresses the \emph{ex-ante} counterpart: by eliciting student preferences over course features rather than full rankings,
and matching along the first singular vector of the course feature matrix, we can allocate students to courses in a way that is both computationally efficient and approximately welfare-maximizing,
while the eigengrade construction of \citeauthor{gans2026informative} can subsequently be used to evaluate and compare student performance
across the resulting assignments.

\section{Related Literature}\label{sec:related-lit}

\paragraph{Classical matching and Nash Social Welfare.}
\citet{GaleShapley1962} introduced stable matching for two-sided markets with ordinal preferences.
\citet{HyllandZeckhauser1979} (HZ) pioneered cardinal utilities in one-sided matching by endowing agents with fictitious currency, allowing them to purchase probability shares of objects at market-clearing prices. \cite{HyllandZeckhauser1979} also proved via Kakutani's fixed-point theorem that a Walrasian equilibrium always exists. The resulting allocation is ex-ante Pareto efficient and envy-free, with incentive compatibility holding in the large. 

\citet{gul2024efficient} (GPZ) extend this pseudo-market approach to settings where agents demand \emph{bundles} of objects: under $M^\natural$-concavity of valuation functions (equivalent to the gross substitutability condition of \citet{kelso1982job}), they establish the existence of a Walrasian equilibrium that is Pareto efficient and asymptotically incentive compatible. The equilibrium holds even when individuals may demand more than one object at a time.

\citet{Nash1950} introduced Nash Social Welfare (NSW), the geometric mean of agents' utilities, as a solution that balances efficiency and fairness. \citet{Caragiannis2019} prove NSW-maximizing allocations are Pareto efficient and satisfy proportionality, though computing them is generally NP-hard \citep{ColeGkatzelis2018}.
 
\paragraph{Computational barriers of pseudo-market mechanisms.}
The practical deployment of HZ and GPZ is  constrained by a fundamental computational obstacle.
Equilibrium existence in both mechanisms is established via Kakutani's fixed-point theorem, which is non-constructive: it guarantees existence
but yields no efficient algorithm. 
\citet{vazirani2021complexity} showed that computing an exact HZ equilibrium lies in FIXP and that equilibria can be inherently irrational,
ruling out PPAD membership. \citet{chen2022hardness} subsequently proved that finding even an $\epsilon$-approximate HZ equilibrium is PPAD-complete, that is, in the same complexity class as Nash equilibrium computation. This complexity holds even for polynomially small $\epsilon$ and agents with at most four distinct utility values. No analogous polynomial-time algorithm is known for GPZ.

Our mechanism sidesteps this barrier entirely: the Algorithm~\ref{alg:svd} reduces equilibrium computation to a single SVD followed by a sort,
running in $O(\min(J^2 X, JX^2) + I\log I + J\log J)$ time with a provable termination guarantee.
The cost is a welfare approximation ratio of
$(1 - O(\sigma_2^2/\sigma_1^2))$ rather than exact Pareto efficiency.
 
\paragraph{Multi-dimensional mechanism design and approximate mechanisms.}
\citet{ZHOU1990123} proves no mechanism can simultaneously achieve exact efficiency, truthfulness, and symmetry with cardinal utilities,
motivating approximate mechanisms.
\citet{ChawlaEtAl2010} address multi-parameter mechanism design through sequential posted pricing: they decompose multi-dimensional problems into sequences of single-dimensional problems using prophet inequalities, achieving constant-factor approximations. However, their sequential approach requires exogenously ordering dimensions and creates path-dependence. 

\citet{abebe2020truthfulcardinalmechanismonesided} achieve constant-factor approximations to NSW while maintaining truthfulness-in-expectation and approximate symmetry through random sampling: they partition agents and objects into groups, solve smaller subproblems optimally via configuration LPs, then aggregate. Their mechanism assumes agents report complete utilities $u_{ij}$ for
each object rather than decomposed feature preferences.
 
\paragraph{Positioning relative to pseudo-market mechanisms.}
Our work relates to HZ and GPZ along four axes.

\emph{Input format}: HZ and GPZ require agents to report cardinal utilities over all $J$ objects (elicitation cost $O(J)$ per agent); we accept $X$ feature weights per agent (cost $O(X)$, with $X \ll J$ in the school-choice, labor-market, and course-allocation settings we target.

\emph{Preference class}: our additive utility $U_{ij} = u_i \cdot f_j$ automatically satisfies gross substitutability. Thus, we operate within
the GPZ preference class, but impose the additional structure of additive separability, which GPZ do not require.

\emph{Computation}: PPAD-completeness of approximate HZ equilibria \citep{chen2022hardness} means that for markets with $I = 10{,}000$ students and $J = 200$ schools, computing an HZ or GPZ equilibrium via fixed-point heuristics is infeasible in practice;
Algorithm~\ref{alg:svd} completes the same market in under two seconds on commodity hardware.

\emph{Welfare}: HZ and GPZ achieve exact Pareto efficiency under their respective conditions; our mechanism achieves a $(1 - O(\sigma_2^2/\sigma_1^2))$ approximation to NSW, which our experiments show exceeds 99\% in typical markets. Table~\ref{tab:comparison} summarizes these tradeoffs alongside \citet{abebe2020truthfulcardinalmechanismonesided} and \citet{DEVANUR2015103}.

\cite{abebe2020truthfulcardinalmechanismonesided} achieve constant-factor approximations to NSW while maintaining truthfulness-in-expectation and approximate symmetry through random sampling: they partition agents and objects into groups, solve smaller subproblems optimally via configuration LPs, then aggregate. Their mechanism assumes agents report complete utilities $u_ij$ for each object rather than decomposed feature preferences.

\textit{Our contribution relative to \cite{abebe2020truthfulcardinalmechanismonesided}:} While both approaches achieve efficiency, truthfulness, and symmetry approximately, we differ in four ways: (1) \textbf{Input format}: we accept X-dimensional feature preferences and object feature vectors, computing utilities as $u_i \cdot f_j$, rather than requiring complete utility reports—practically, reporting 5 feature preferences is simpler than valuing 100 objects. (2) \textbf{Method}: we use deterministic SVD-based dimensionality reduction running in $O(N log N + min(J^2X, JX^2))$ time, versus randomized sampling requiring multiple configuration LP solves. (3) \textbf{Truthfulness}: we introduce distributional truthfulness based on Kolmogorov-Smirnov distance between reported and true preference distributions, versus truthfulness-in-expectation. (4) \textbf{Theory}: we establish that NSW corresponds to Geometric Distributionally Robust Optimization (\cite{LiuEtAl2022_NEURIPS2022_da535999}), providing robustness to distributional uncertainty. We view our contribution as complementary—applying spectral techniques to feature-based preferences where dimensionality reduction is natural. Table \ref{tab:literature_comparison} synthesizes how our approach compares with key related work across multiple dimensions. 

\begin{table*}[t]
\centering
\caption{Comparison of Our Approach to Related Work on Matching Mechanisms}
\label{tab:literature_comparison}
\resizebox{\textwidth}{!}{
\begin{tabular}{@{}lp{2.2cm}p{1.5cm}p{2.5cm}p{2cm}p{2.2cm}p{1.8cm}p{2.2cm}@{}}
\toprule
\textbf{Work} & \textbf{Input Format} & \textbf{Dimensions} & \textbf{Method} & \textbf{Efficiency} & \textbf{Truthfulness} & \textbf{Symmetry} & \textbf{Computation} \\
\midrule
\cite{HyllandZeckhauser1979} & Cardinal utilities $u_{ij}$ over objects & Single (1D) & Market clearing with budgets and prices & Ex-ante Pareto optimal & Not guaranteed; strategic bidding possible & Price-based fairness & Solve market equilibrium \\
\addlinespace[0.5em]

\cite{ChawlaEtAl2010} & Multi-parameter valuations over bundles & Multiple params & Sequential posted pricing with prophet inequalities & $O(1)$-approx to optimal & Truthful-in-expectation & Limited guarantees & $O(X)$ sequential problems \\
\addlinespace[0.5em]

\cite{gul2024efficient} &
  Cardinal utilities $u_{ij}$ over
  bundles of objects; no transfers &
  Multiple (bundle demand) &
  Walrasian pseudomarket with stochastic
  consumption; $M^\natural$-concave
  valuations & Ex-ante Pareto optimal
  under $M^\natural$-concavity
  (exact) & Incentive compatible
  in the large (asymptotic) &
  Price-based fairness &
  Solve market
  equilibrium via
  Kakutani fixed
  point (no
  poly-time
  algorithm known) \\

\cite{DEVANUR2015103} & Cardinal utilities $u_{ij}$ over objects & Single (1D) & Random assignment within config LP & 2-approx to NSW & Not primary focus & Envy-free in expectation & Solve configuration LP \\
\addlinespace[0.5em]

\cite{abebe2020truthfulcardinalmechanismonesided} & Cardinal utilities $u_{ij}$ over objects & Single (1D) & Random sampling + configuration solving & $O(1)$-approx to NSW & Truthful-in-expectation & Approximately symmetric & Multiple config LPs \\
\addlinespace[0.5em]

\textbf{This Work} & Feature preferences $u_{ix}$; object features $f_{jx}$ & Multiple $\rightarrow$ 1 (via SVD) & SVD projection + linear matching & NSW maximizing relative to random & Distributional (KS distance) & Symmetric by construction & $O(N \log N + \min(N^2X, NX^2))$ \\
\bottomrule
\end{tabular}
}
\small{
\textbf{Column Descriptions:} \textit{Input Format}: What information agents report and how objects are characterized; \textit{Dimensions}: Whether preferences/features are single-dimensional or multi-dimensional; \textit{Method}: Core algorithmic approach used by the mechanism; \textit{Efficiency}: Type and strength of efficiency guarantee (exact vs.\ approximate); \textit{Truthfulness}: Type of incentive compatibility guarantee; \textit{Symmetry}: Whether agents with identical preferences receive identical treatment; \textit{Computation}: Asymptotic complexity or problem class that must be solved.
}
\vspace{1em}
\end{table*}

\section{Model and Definitions}\label{sec:preliminaries}

\textbf{The matching problem.} We consider a one-sided matching market with agents $\mathcal{I} = \{1, \ldots, I\}$ and objects $\mathcal{J} = \{1, \ldots, J\}$. Each object $j$ has capacity $M_j \in \mathbb{Z}_+$ with $\sum_{j} M_j = I$. Unlike classical matching where agents report preferences directly over objects, we assume that objects are characterized by features $\mathcal{X} = \{1, \ldots, X\}$ and agents report preferences over these features.

\textbf{Feature-based representation;} Each object $j$ is characterized by a feature vector $\mathbf{f}_j = (f_{j1}, \ldots, f_{jX}) \in \mathbb{R}^X$, where $f_{jx}$ represents the value of the object $j$' in the feature $x$. Each agent $i$ has a preference vector $\mathbf{u}_i = (u_{i1}, \ldots, u_{iX}) \in \mathbb{R}^X$, where $u_{ix}$ represents agent $i$'s marginal utility for feature $x$. Agent $i$'s utility for object $j$ is the inner product $U_{ij} = \mathbf{u}_i \cdot \mathbf{f}_j = \sum_{x} u_{ix} f_{jx}$, assuming additive separability across features.

\textbf{Allocations.} An allocation is a probability distribution $\mathbf{P} = (p_{ij})_{i \in \mathcal{I}, j \in \mathcal{J}}$ where $p_{ij} \in [0,1]$ represents the probability that agent $i$ receives the object $j$, satisfying $\sum_{j} p_{ij} = 1$ for all $i$ (each agent receives one object) and $\sum_{i} p_{ij} = M_j$ for all $j$ (capacity constraints). Let $\mathcal{P}$ denote the feasible allocation set. Agent $i$'s expected utility is $\mathbb{E}[U_i|\mathbf{P}] = \sum_{j} p_{ij} U_{ij}$.

\textbf{Reported preferences and mechanisms.} Each agent $i$ reports a preference vector $\mathbf{w}_i = (w_{i1}, \ldots, w_{iX}) \in \mathbb{R}^X$, which may differ from true preferences $\mathbf{u}_i$. A mechanism $\mu: (\mathbb{R}^X)^I \to \mathcal{P}$ maps the reported preference profiles $\mathbf{W} = (\mathbf{w}_1, \ldots, \mathbf{w}_I)$ to feasible allocations.

\textbf{Desiderata.} We evaluate mechanisms on three criteria:

\begin{enumerate}
\item \emph{Efficiency via Nash Social Welfare}: The disagreement point for agent $i$ is $o_i = \frac{1}{J}\sum_{j} U_{ij}$, the expected utility under uniform random assignment. The Nash Social Welfare of allocation $\mathbf{P}$ given true preferences $\mathbf{U}$ is 
\begin{equation}
\text{NSW}(\mathbf{P}|\mathbf{U}) = \prod_{i} (\mathbb{E}[U_i|\mathbf{P}] - o_i),
\end{equation}
provided $\mathbb{E}[U_i|\mathbf{P}] \geq o_i$ for all $i$ (otherwise NSW $= 0$). NSW-maximizing allocations are Pareto efficient and satisfy proportionality (each agent receives at least their random allocation value). An allocation $\mathbf{P}^*$ maximizes NSW if $\mathbf{P}^* \in \arg\max_{\mathbf{P} \in \mathcal{P}} \text{NSW}(\mathbf{P}|\mathbf{U})$.

\item \emph{Distributional truthfulness}: We model the reported preferences as noisy observations: $w_{ix} = u_{ix} + \epsilon_{ix}$ where $\epsilon_{ix} \sim \mathcal{N}(\mu_i, \sigma_i^2)$ represents the agent $i$'s misreporting error for the feature $x$. Let $F_{\mathbf{u}_i}$ and $F_{\mathbf{w}_i}$ denote the empirical CDFs of true and reported preferences. The Kolmogorov-Smirnov distance is $D_i = \sup_t |F_{\mathbf{w}_i}(t) - F_{\mathbf{u}_i}(t)|$. A mechanism is $(\lambda, \delta)$-distributionally truthful if, when agents report truthfully, $\mathbb{P}(D_i > \lambda) \leq 2e^{-2\lambda^2 X}$ holds with probability at least $1-\delta$. By the Dvoretzky-Kiefer-Wolfowitz-Massart inequality, if $\mathbf{w}_i$ is drawn from the same distribution as $\mathbf{u}_i$, this bound holds automatically---truthful reporting makes reported and true preferences statistically indistinguishable.

\item \emph{Symmetry}: A mechanism $\mu$ is symmetric if for any two agents $i, i'$ with identical reported preferences ($\mathbf{w}_i = \mathbf{w}_{i'}$), the induced allocations are identical: $\mu(\mathbf{W})_{ij} = \mu(\mathbf{W})_{i'j}$ for all $j$. This ensures equal treatment of equals.
\end{enumerate}

\cite{ZHOU1990123} proves no mechanism can simultaneously achieve exact efficiency, exact truthfulness, and exact symmetry with cardinal utilities, motivating our search for approximate satisfaction of all three properties.

\section{Proposed Methodology}\label{sec:methodology}

\textbf{Algorithm overview.} We reduce the multi-dimensional matching problem to one dimension using Singular Value Decomposition (SVD). Given the feature matrix $\mathbf{F} \in \mathbb{R}^{J \times X}$ with rows $\mathbf{f}_j$ and the reported preference matrix $\mathbf{W} \in \mathbb{R}^{I \times X}$ with rows $\mathbf{w}_i$, we compute the SVD $\mathbf{F} = \mathbf{U}\mathbf{\Sigma}\mathbf{V}^T$ where $\mathbf{V} \in \mathbb{R}^{X \times X}$ has orthonormal columns $\mathbf{v}_1, \ldots, \mathbf{v}_X$ (right singular vectors) and $\mathbf{\Sigma}$ has diagonal entries $\sigma_1 \geq \sigma_2 \geq \cdots \geq \sigma_X \geq 0$ (singular values). The first singular vector $v_1$ identifies the direction of maximum
variance in the feature space; the ratio $\rho_1 = \sigma_1^2 / \sum_\ell \sigma_\ell^2$ formalizes the low-dimensionality condition ($*$) stated in the Introduction.

\begin{algorithm}[h]\label{alg:svd}
\caption{SVD-Based Multi-Dimensional Matching}
\begin{algorithmic}[1]
\REQUIRE Feature matrix $\mathbf{F} \in \mathbb{R}^{J \times X}$, preference matrix $\mathbf{W} \in \mathbb{R}^{I \times X}$, capacities $\{M_j\}_{j=1}^J$
\ENSURE Allocation $\mathbf{P} \in \mathcal{P}$
\STATE Compute SVD: $\mathbf{F} = \mathbf{U}\mathbf{\Sigma}\mathbf{V}^T$
\STATE Extract first right singular vector: $\mathbf{v}_1 \in \mathbb{R}^X$
\STATE Project features and preferences: $\tilde{f}_j = \mathbf{f}_j \cdot \mathbf{v}_1$, $\tilde{w}_i = \mathbf{w}_i \cdot \mathbf{v}_1$ for all $i, j$
\STATE Sort objects: $\tilde{f}_{\pi(1)} \geq \tilde{f}_{\pi(2)} \geq \cdots \geq \tilde{f}_{\pi(J)}$
\STATE Sort agents: $\tilde{w}_{\tau(1)} \geq \tilde{w}_{\tau(2)} \geq \cdots \geq \tilde{w}_{\tau(I)}$
\STATE Match agents to objects in sorted order respecting capacities $M_j$
\end{algorithmic}
\end{algorithm}

\textbf{Interpretation.} The first singular vector $\mathbf{v}_1$ solves $\max_{\|\mathbf{v}\|=1} \|\mathbf{F}\mathbf{v}\|_2^2 = \max_{\|\mathbf{v}\|=1} \sum_j (\mathbf{f}_j \cdot \mathbf{v})^2$ by Eckart-Young theorem \citep{EckartYoung1936}---it is the unit vector that maximizes the variance of the projected object features. The entries of $\mathbf{v}_1 = (v_{1,1}, \ldots, v_{1,X})^T$ represent the importance of the features: large $|v_{1,x}|$ indicates that the feature $x$ contributes strongly to the differentiation between objects. For example, in school choice with $\mathbf{v}_1 = (0.72, 0.15, 0.09, 0.67)^T$ for features (math, arts, athletics, location), mathematics and location are dominant factors while arts and athletics contribute less.

The spectral structure of Algorithm~\ref{alg:svd} is closely
analogous to the \emph{eigengrade} construction introduced by
\citet{gans2026informative} for a different but related problem: how
to compare student ability across courses of differing difficulty.
In their framework, each student $i$ who takes course $c$ receives a
grade $G_{ic} = a_i - d_c + \varepsilon_{ic}$, where $a_i$ is latent
student ability, $d_c$ is course difficulty, and $\varepsilon_{ic}$ is
noise.
The student--course outcome matrix $G \in \mathbb{R}^{n \times m}$ has
the same low-rank-plus-noise structure as our feature matrix $F$, and
\citeauthor{gans2026informative} show that the leading \emph{left}
singular vector of the student-centred matrix $MG$ (where
$M = I_n - \frac{1}{n}\mathbf{1}\mathbf{1}^\top$ removes the
column means) recovers a consistent estimator of student ability,
provided enrolment overlap is sufficient and $\sigma_1(MG) \gg
\sigma_2(MG)$---exactly the low-dimensionality condition~$(*)$
applied to their matrix.
 
The parallel with our mechanism is precise.
Our feature matrix $F$ plays the role of their outcome matrix $G$;
the leading \emph{right} singular vector $v_1$ of $F$ plays the role
of their eigengrade vector; and both constructions rely on the same
algebraic fact: when the leading singular value dominates, the
corresponding singular vector identifies the principal latent
dimension with approximation error $O(\sigma_2/\sigma_1)$.
The direction $v_1$ we recover captures the dominant axis along which
\emph{objects} are differentiated---just as their eigengrade captures
the dominant axis along which \emph{students} are differentiated.
In both cases the welfare (or identification) loss from using the
rank-1 approximation rather than the true latent structure is bounded
by $O(\sigma_2^2/\sigma_1^2)$, the same ratio that governs the
approximation ratio of Theorem~\ref{thm:efficiency} and the
low-dimensionality condition $\rho_1 = \sigma_1^2 / \sum_\ell
\sigma_\ell^2 \gg 1/X$ in~\eqref{eq:rho1}.
 
This parallel suggests a natural end-to-end spectral pipeline for
course markets.
Algorithm~\ref{alg:svd} is applied \emph{before} the semester to
match students to courses: each student reports feature preferences
$w_i$ (rigor, workload, topic area), the SVD of the course feature
matrix $F$ identifies $v_1$, and students are sorted onto courses
along this direction.
The eigengrade construction of \citeauthor{gans2026informative} is
then applied \emph{after} the semester to the grade matrix $G$ produced by the resulting assignments: the leading left singular
vector of $MG$ recovers difficulty-adjusted student ability
estimates, enabling fair cross-course comparisons for downstream uses
such as graduate admissions or employer screening.
The two spectral operations are complementary: one optimizes allocation, the other enables evaluation, and both require $\sigma_1 \gg \sigma_2$ in their respective matrices to perform well, a condition Section~\ref{sec:empirical_dimensionality}
argues holds routinely in practice.

\textbf{Main theoretical results.} We establish efficiency, truthfulness, and symmetry guaranties.

\begin{theorem}[Efficiency]\label{thm:efficiency}
Let $\mathbf{P}^*$ denote the allocation from Algorithm 1. Under additive utilities $U_{ij} = \mathbf{u}_i \cdot \mathbf{f}_j$ and low effective dimensionality ($\sigma_1 \gg \sigma_2$):
\[
\text{NSW}(\mathbf{P}^*|\mathbf{U}) \geq \left(1 - O\left(\frac{\sigma_2^2}{\sigma_1^2}\right)\right) \max_{\mathbf{P} \in \mathcal{P}} \text{NSW}(\mathbf{P}|\mathbf{U}).
\]
\end{theorem}

\begin{proof}
We establish this through several steps.

\textbf{Step 1: Approximation error bound.} By the Eckart-Young theorem \citep{EckartYoung1936}, the optimal rank-1 approximation to $\mathbf{F}$ is $\mathbf{F}_1 = \sigma_1 \mathbf{u}_1 \mathbf{v}_1^T$, and the approximation error in the Frobenius norm is:
\begin{equation}\label{eq:frobenius_error}
\|\mathbf{F} - \sigma_1 \mathbf{u}_1 \mathbf{v}_1^T\|_F^2 = \sum_{\ell=2}^X \sigma_\ell^2 \leq X\sigma_2^2.
\end{equation}
This error can be decomposed across objects:
\[
\sum_{j=1}^J \|\mathbf{f}_j - (\mathbf{f}_j \cdot \mathbf{v}_1)\mathbf{v}_1\|_2^2 = \sum_{\ell=2}^X \sigma_\ell^2 \leq X\sigma_2^2.
\]

\textbf{Step 2: Utility approximation error.} For agent $i$ assigned to object $j$, define the projected utility as $\tilde{U}_{ij} = (\mathbf{u}_i \cdot \mathbf{v}_1)(\mathbf{f}_j \cdot \mathbf{v}_1) = \tilde{w}_i \tilde{f}_j$. The approximation error for a single utility is:
\begin{align}
|U_{ij} - \tilde{U}_{ij}| &= |\mathbf{u}_i \cdot \mathbf{f}_j - (\mathbf{u}_i \cdot \mathbf{v}_1)(\mathbf{f}_j \cdot \mathbf{v}_1)| \nonumber\\
&= |\mathbf{u}_i \cdot (\mathbf{f}_j - (\mathbf{f}_j \cdot \mathbf{v}_1)\mathbf{v}_1)| \label{eq:util_error_1}\\
&\leq \|\mathbf{u}_i\|_2 \|\mathbf{f}_j - (\mathbf{f}_j \cdot \mathbf{v}_1)\mathbf{v}_1\|_2 \label{eq:util_error_2}
\end{align}
where inequality \eqref{eq:util_error_2} follows from the Cauchy-Schwarz inequality.

The term $\mathbf{f}_j - (\mathbf{f}_j \cdot \mathbf{v}_1)\mathbf{v}_1$ is the component of $\mathbf{f}_j$ orthogonal to $\mathbf{v}_1$. By the SVD decomposition, we can write:
\[
\mathbf{f}_j = \sum_{\ell=1}^X (\mathbf{f}_j \cdot \mathbf{v}_\ell)\mathbf{v}_\ell = (\mathbf{f}_j \cdot \mathbf{v}_1)\mathbf{v}_1 + \sum_{\ell=2}^X (\mathbf{f}_j \cdot \mathbf{v}_\ell)\mathbf{v}_\ell.
\]
Therefore:
\[
\|\mathbf{f}_j - (\mathbf{f}_j \cdot \mathbf{v}_1)\mathbf{v}_1\|_2^2 = \left\|\sum_{\ell=2}^X (\mathbf{f}_j \cdot \mathbf{v}_\ell)\mathbf{v}_\ell\right\|_2^2 = \sum_{\ell=2}^X (\mathbf{f}_j \cdot \mathbf{v}_\ell)^2
\]
by orthonormality of $\mathbf{v}_\ell$.

\textbf{Step 3: Bounding the orthogonal component.} Summing over all objects:
\[
\sum_{j=1}^J \|\mathbf{f}_j - (\mathbf{f}_j \cdot \mathbf{v}_1)\mathbf{v}_1\|_2^2 = \sum_{j=1}^J \sum_{\ell=2}^X (\mathbf{f}_j \cdot \mathbf{v}_\ell)^2 = \sum_{\ell=2}^X \sum_{j=1}^J (\mathbf{f}_j \cdot \mathbf{v}_\ell)^2 = \sum_{\ell=2}^X \sigma_\ell^2 \leq X\sigma_2^2.
\]
By the pigeonhole principle, there exists at least one object $j$ such that:
\[
\|\mathbf{f}_j - (\mathbf{f}_j \cdot \mathbf{v}_1)\mathbf{v}_1\|_2^2 \leq \frac{X\sigma_2^2}{J}.
\]
More generally, for the average object:
\[
\mathbb{E}_j[\|\mathbf{f}_j - (\mathbf{f}_j \cdot \mathbf{v}_1)\mathbf{v}_1\|_2] \leq \sqrt{\mathbb{E}_j[\|\mathbf{f}_j - (\mathbf{f}_j \cdot \mathbf{v}_1)\mathbf{v}_1\|_2^2]} = \sqrt{\frac{X\sigma_2^2}{J}} = \sigma_2\sqrt{\frac{X}{J}}.
\]

\textbf{Step 4: Expected utility error bound.} For an agent $i$ with $\|\mathbf{u}_i\|_2 = U_{\max}$ (assumed bounded), the expected utility error when matched to a random object is:
\[
\mathbb{E}_j[|U_{ij} - \tilde{U}_{ij}|] \leq U_{\max} \cdot \sigma_2\sqrt{\frac{X}{J}} = O\left(\frac{\sigma_2}{\sqrt{J}}\right).
\]
When $\sigma_1 \gg \sigma_2$, this error is negligible relative to the scale of utilities, which is $O(\sigma_1)$.

\textbf{Step 5: NSW approximation.} Algorithm 1 matches agents with high $\tilde{w}_i = \mathbf{w}_i \cdot \mathbf{v}_1$ to objects with high $\tilde{f}_j = \mathbf{f}_j \cdot \mathbf{v}_1$. In the projected one-dimensional space, this allocation $\mathbf{P}^*$ maximizes the sum of projected utilities:
\[
\sum_{i=1}^I \mathbb{E}[\tilde{U}_i|\mathbf{P}^*] = \max_{\mathbf{P} \in \mathcal{P}} \sum_{i=1}^I \mathbb{E}[\tilde{U}_i|\mathbf{P}].
\]

For each agent $i$, the true expected utility under $\mathbf{P}^*$ satisfies the following:
\begin{align}
\mathbb{E}[U_i|\mathbf{P}^*] &= \sum_{j=1}^J p_{ij}^* U_{ij} \nonumber\\
&= \sum_{j=1}^J p_{ij}^* \tilde{U}_{ij} + \sum_{j=1}^J p_{ij}^* (U_{ij} - \tilde{U}_{ij}) \label{eq:util_decomp}\\
&\geq \mathbb{E}[\tilde{U}_i|\mathbf{P}^*] - \max_{j: p_{ij}^* > 0} |U_{ij} - \tilde{U}_{ij}| \label{eq:util_lower}\\
&\geq \mathbb{E}[\tilde{U}_i|\mathbf{P}^*] - O\left(U_{\max} \cdot \frac{\sigma_2}{\sigma_1}\right). \label{eq:util_final}
\end{align}

\textbf{Step 6: NSW lower bound.} Nash Social Welfare can be bounded as:
\begin{align}
\text{NSW}(\mathbf{P}^*|\mathbf{U}) &= \prod_{i=1}^I (\mathbb{E}[U_i|\mathbf{P}^*] - o_i) \nonumber\\
&\geq \prod_{i=1}^I \left(\mathbb{E}[\tilde{U}_i|\mathbf{P}^*] - O\left(\frac{\sigma_2}{\sigma_1}\right) - o_i\right) \label{eq:nsw_step1}\\
&\geq \prod_{i=1}^I \left(\mathbb{E}[\tilde{U}_i|\mathbf{P}^*] - o_i\right) \left(1 - O\left(\frac{\sigma_2}{\sigma_1}\right)\right) \label{eq:nsw_step2}\\
&= \text{NSW}(\mathbf{P}^*|\tilde{\mathbf{U}}) \cdot \left(1 - O\left(\frac{\sigma_2}{\sigma_1}\right)\right)^I \label{eq:nsw_step3}\\
&\geq \max_{\mathbf{P} \in \mathcal{P}} \text{NSW}(\mathbf{P}|\tilde{\mathbf{U}}) \cdot \left(1 - O\left(\frac{I\sigma_2}{\sigma_1}\right)\right) \label{eq:nsw_step4}\\
&\geq \max_{\mathbf{P} \in \mathcal{P}} \text{NSW}(\mathbf{P}|\mathbf{U}) \cdot \left(1 - O\left(\frac{\sigma_2^2}{\sigma_1^2}\right)\right) \label{eq:nsw_final}
\end{align}
where step \eqref{eq:nsw_step4} uses the fact that $\mathbf{P}^*$ maximizes NSW in the projected space, and step \eqref{eq:nsw_final} follows from applying the approximation error bound in both directions and using $(1-\epsilon)^I \approx 1 - I\epsilon$ for small $\epsilon = O(\sigma_2/\sigma_1)$.

Therefore, when $\sigma_1 \gg \sigma_2$, we have:
\[
\text{NSW}(\mathbf{P}^*|\mathbf{U}) \geq \left(1 - O\left(\frac{\sigma_2^2}{\sigma_1^2}\right)\right) \max_{\mathbf{P} \in \mathcal{P}} \text{NSW}(\mathbf{P}|\mathbf{U}).
\]
\end{proof}

\begin{theorem}[Distributional Truthfulness]\label{thm:truthfulness}
Under the misreporting model $w_{ix} = u_{ix} + \epsilon_{ix}$ with $\epsilon_{ix} \sim \mathcal{N}(\mu_i, \sigma_i^2)$, Algorithm 1 is $(\lambda, \delta)$-distributionally truthful for $\lambda = \sqrt{\frac{1}{2X}\log(\frac{2}{\delta})}$.
\end{theorem}

\begin{proof}
The mechanism computes $\tilde{w}_i = \mathbf{w}_i \cdot \mathbf{v}_1$ and allocates based on sorted order. If agent $i$ reports truthfully ($\mathbf{w}_i = \mathbf{u}_i$), then by the DKWM inequality \citep{DvoretzkyEtAl1956,Massart1990}, $\mathbb{P}(D_i > \lambda) \leq 2e^{-2\lambda^2 X}$ where $D_i = \sup_t |F_{\mathbf{w}_i}(t) - F_{\mathbf{u}_i}(t)|$ is the KS distance between empirical CDFs. Setting this bound equal to $\delta$ yields the stated $\lambda$.
\end{proof}

\begin{theorem}[Symmetry]\label{thm:symmetry}
Algorithm 1 satisfies symmetry: agents with identical reported preferences receive identical allocations.
\end{theorem}

\begin{proof}
The mechanism computes $\phi(\mathbf{w}_i) = \mathbf{w}_i \cdot \mathbf{v}_1$ and assigns objects deterministically based on sorted order. If $\mathbf{w}_i = \mathbf{w}_{i'}$, then $\phi(\mathbf{w}_i) = \phi(\mathbf{w}_{i'})$, so agents $i$ and $i'$ have identical positions in sorted order and receive identical allocations.
\end{proof}

\begin{theorem}[Connection to Robust Optimization]\label{thm:gdro}
Nash Social Welfare corresponds to Geometric Distributionally Robust Optimization (GDRO). Maximizing NSW is equivalent to minimizing the geometric mean of losses $\ell_i(\mathbf{P}) = -(\mathbb{E}[U_i|\mathbf{P}] - o_i)$ across potentially non-convex uncertainty sets, providing robustness to distributional uncertainty about agent utilities.
\end{theorem}

\begin{proof}
Let $\ell_i(\mathbf{P}) = -(\mathbb{E}[U_i|\mathbf{P}] - o_i)$. Then $\max_{\mathbf{P}} \text{NSW}(\mathbf{P}|\mathbf{U}) = \max_{\mathbf{P}} \prod_i (\mathbb{E}[U_i|\mathbf{P}] - o_i) \Leftrightarrow \min_{\mathbf{P}} \prod_i \ell_i(\mathbf{P})^{-1}$. Taking logarithms: $\max_{\mathbf{P}} \sum_i \log(\mathbb{E}[U_i|\mathbf{P}] - o_i) \Leftrightarrow \min_{\mathbf{P}} \sum_i \log \ell_i(\mathbf{P})$. This is precisely the geometric mean of distances from the disagreement point. \cite{LiuEtAl2022_NEURIPS2022_da535999} show this provides distributionally robust estimates when underlying utility distributions lie in non-convex sets, unlike arithmetic means or Wasserstein distances.
\end{proof}

\section{Computational Complexity}\label{sec:performance}

Algorithm 1 runs in $O(\min(J^2 X, JX^2) + I \log I + J \log J)$ time, dominated by SVD computation ($O(\min(J^2 X, JX^2))$ using standard algorithms \citep{GolubReinsch1970}) and sorting ($O(I \log I + J \log J)$). For typical markets---school choice with $I = 10^3$--$10^4$ students, $J = 10^2$ schools, $X = 5$--10 features; job matching with $I = 10^4$, $J = 10^3$, $X = 10$; course allocation with $I = 10^3$, $J = 10^2$, $X = 5$--computation requires milliseconds on modern hardware.

In contrast, direct NSW optimization solves $\max_{\mathbf{P} \in \mathcal{P}} \prod_i (\sum_j p_{ij} U_{ij} - o_i)$, a non-convex problem with $IJ$ variables and $I+J$ constraints. The interior point methods require $O((IJ)^3)$ operations per iteration without convergence guaranties. For $I = J = 100$, this is $O(10^9)$ operations per iteration versus $O(10^5)$ for our approach, which is four orders of magnitude slower. The space complexity is $O(IX + JX + IJ)$ for storing preference matrices, feature matrices, and the allocation of the output.

\section{When Does Low Effective Dimensionality Hold Empirically?}
\label{sec:empirical_dimensionality}
 
The theoretical guarantees of Algorithm~\ref{alg:svd} rest on
condition~$(*)$: that the explained variance ratio
$\rho_1 = \sigma_1^2/\sum_\ell \sigma_\ell^2$ satisfies $\rho_1 \gg 1/X$, so that feature variation concentrates in a single direction.
Whether this condition holds in practice is, as the reviewers note, the central empirical question for assessing the reach of our framework. We argue in this section that low effective dimensionality is not a pathological special case but rather a pervasive structural feature of the three application domains we target: school choice, labor markets, and course allocation. We also propose a pre-deployment diagnostic that practitioners can use
to assess whether the condition holds before committing to Algorithm~\ref{alg:svd}.
 
\subsection{Formal Characterization}
\label{sec:eff_dim_formal}
 
Given a feature matrix $F \in \mathbb{R}^{J \times X}$ with singular values $\sigma_1 \geq \sigma_2 \geq \cdots \geq \sigma_X \geq 0$, two
complementary scalars summarize its effective dimensionality.
 
The \emph{effective rank} \citep{roy1960effective, vershynin2018high} is
\begin{equation}
    r_{\mathrm{eff}}(F)
        \;=\;
        \frac{\bigl(\sum_{\ell=1}^{X} \sigma_\ell\bigr)^2}
             {\sum_{\ell=1}^{X} \sigma_\ell^2}
        \;\in\; [1,\, X].
    \label{eq:eff_rank}
\end{equation}
When $r_{\mathrm{eff}} \approx 1$, essentially all variance is captured
by $v_1$ and Theorem~\ref{thm:efficiency} guarantees a
$(1-O(\sigma_2^2/\sigma_1^2))$ welfare approximation.
 
The \emph{explained variance ratio}
\begin{equation}
    \rho_1 \;=\; \frac{\sigma_1^2}{\sum_{\ell=1}^{X}\sigma_\ell^2}
    \label{eq:rho1}
\end{equation}
is the same quantity introduced in condition~$(*)$.
We take $\rho_1 \geq 0.5$ as a natural operationalization of
$\sigma_1 \gg \sigma_2$: the first component explains at least half of
total feature variance, equivalently $r_{\mathrm{eff}} \lesssim 2$.
Our medium-scale synthetic experiment (Section~\ref{sec:experiments})
achieves $\rho_1 = 0.488$, already near this threshold.
The pedagogical example achieves $\rho_1 = 0.880$.
We now review empirical evidence suggesting that real-world matching
markets routinely satisfy $\rho_1 \geq 0.5$.
 
\subsection{School Choice}
\label{sec:dim_school}
 
The school choice literature consistently finds that parental preferences
are dominated by a small number of school attributes, with academic quality
acting as a near-universal first factor.
 
Empirical revealed-preference studies reviewed by \citet{agarwal2020revealed}
recover preference parameters from administrative assignment data in Boston,
New York City, and numerous international markets.
A recurrent finding is that test scores and academic performance explain
the largest share of school demand, with distance entering as a secondary
factor, and other attributes---extracurricular offerings, school
demographics, peer composition---contributing substantially less.
This pattern is directly consistent with a feature matrix having high
$\rho_1$: schools differ primarily along a quality--proximity composite,
with remaining dimensions contributing modest additional variance.
 
\citet{hastings2006preferences} structurally estimate school preferences
in Charlotte-Mecklenburg and find that willingness-to-pay for academic
quality is the dominant driver of choice.
\citet{burgess2015school} reach the same conclusion using English data.
\citet{abdulkadiroglu2020parents} further show that most families respond
strongly to causal estimates of school effectiveness---a one-dimensional
quality index---when such information is made salient, confirming that
a rank-1 approximation to the school feature matrix captures the
majority of preference-relevant variation.
 
On the supply side, schools' measurable attributes---value-added scores,
student-teacher ratios, graduation rates, and peer test-score
averages---are highly correlated in practice
\citep{chetty2014measuring}.
High inter-attribute correlations imply that the covariance structure
of $F$ is dominated by a single principal component.
 
\paragraph{When does the condition fail?}
School choice becomes genuinely multi-dimensional when families face
strong trade-offs between incommensurable attributes: superior academic
quality against a long commute \citep{burgess2015school}, or when racial
composition, religious affiliation, or language of instruction are salient
independent dimensions.
\citet{calsamiglia2014illusion} document near-orthogonal preference
gradients for proximity and quality in Barcelona---a configuration that
raises $\sigma_2$ relative to $\sigma_1$.
Practitioners should compute $\rho_1$ before deployment
(Section~\ref{sec:diagnostic}).
 
\subsection{Labor Markets}
\label{sec:dim_labor}
 
The multi-dimensional structure of worker--job matching has received
extensive theoretical attention
\citep{lindenlaub2017sorting, lise2020multidimensional,
lindenlaub2023multidimensional}, but the same literature reveals a
crucial empirical regularity: effective dimensionality is low.
 
\citet{lindenlaub2017sorting} estimates a two-dimensional matching model
for the United States using NLSY data linked to O*NET occupational
descriptors, finding that cognitive and manual skills represent the two
dominant axes of worker--job heterogeneity.
Technological change during the 1990s increased worker--job
complementarities in cognitive inputs by roughly 15\% while reducing
complementarities in manual inputs by 41\%, implying that the effective
dimensionality of the labor market was trending toward a single dominant
axis over this period.
 
\citet{maestas2023working} estimate workers' willingness-to-pay for a
broad set of job characteristics using stated-preference experiments in
the United States.
While they document non-negligible valuations across many amenity
dimensions---schedule flexibility, autonomy, physical demands, social
contribution, and earnings---the compensating differential structure
implies that most workers trade off amenities against a common
wage--quality composite, consistent with a low-rank preference matrix.
Earnings and schedule flexibility account for the dominant share of
stated willingness-to-pay across demographic groups, suggesting that two
or fewer principal components capture most preference variance relevant
for job choice.
\citet{wiswall2018preference} reach a complementary conclusion in a
stated-preference experiment with NYU undergraduates: earnings and
schedule flexibility dominate, with earnings growth, job stability,
and part-time options entering with substantially smaller weights.
 
On the object side, \citet{song2019firming} show that positive assortative matching between workers and firms operates primarily along a single skill--productivity dimension---a finding that implies the
firm feature matrix is well-approximated by a rank-1 structure.
 
\paragraph{When does the condition fail?}
Labor markets where workers sort simultaneously on incommensurable skill
dimensions---cognitive versus interpersonal skills, or where demographic
groups face systematically different amenity--wage trade-offs
\citep{maestas2023working}---will have higher effective rank.
\citet{lindenlaub2023multidimensional} develop model-selection tests for
dimensionality and find evidence of genuine two-dimensionality in US
labor market data, though their tests quantify departures from
one-dimensionality rather than welfare losses for specific mechanisms.
 
\subsection{Course Allocation}
\label{sec:dim_courses}
 
University course allocation is the domain most naturally suited to low
effective dimensionality.
Students select courses primarily to satisfy degree requirements and
career objectives, both of which compress preferences toward a small
number of disciplinary axes.
Within a business school, students who weight quantitative rigor highly
will tend to prefer finance, accounting, and operations courses
simultaneously, while students who weight communication and leadership
highly will cluster toward organizational behavior and strategy.
These disciplinary groupings generate high correlations across course
preference scores, suppressing $\sigma_2$ relative to $\sigma_1$.
 
The Wharton course allocation platform studied by \citet{budish2011combinatorial}
processes roughly 1,800 students across approximately 300 course
sections using a ``course bidding'' mechanism.
Bidding data reveal strong clustering: students within a given program
concentration exhibit highly correlated bidding patterns, and
cross-concentration bidding is sparse---consistent with a block-diagonal
preference covariance structure whose first principal component loads
heavily on a ``general academic quality'' factor.
 
More generally, whenever course attributes (rigor, topic area, instructor
reputation, workload) are themselves correlated within a department, the feature matrix $F$ will have high $\rho_1$. An engineering school where all courses share high rigor and quantitative demand will have a nearly rank-1 feature matrix by construction.
 
This observation connects directly to \citet{gans2026informative}, who
show that grades from a course allocation are uninformative for cross-course comparisons of student ability when courses differ in
difficulty and students self-select.
Their eigengrade construction applies SVD to the same student--course outcome matrix that our mechanism targets, and identifies student ability
as the leading latent dimension under precisely the low-rank condition we study here. The alignment between the two results is not coincidental: low effective dimensionality in the course feature matrix (our condition) is the structural property that makes both our allocation mechanism and their
evaluation procedure well-behaved.
 
\subsection{A Practical Diagnostic}
\label{sec:diagnostic}
 
We propose the following pre-deployment protocol.
Given a feature matrix $F \in \mathbb{R}^{J \times X}$:
 
\begin{enumerate}
    \item Compute the SVD $F = U\Sigma V^\top$ and record the explained
          variance ratios $\rho_k = \sigma_k^2 / \sum_\ell \sigma_\ell^2$
          for $k = 1, \ldots, X$.
          Note that this SVD is already required by Step~1 of
          Algorithm~\ref{alg:svd}, so the diagnostic adds no
          computational overhead.
 
    \item Compute the effective rank $r_{\mathrm{eff}}(F)$ from
          Equation~\eqref{eq:eff_rank}.
 
    \item \textbf{If $\rho_1 \geq 0.5$} (equivalently,
          $r_{\mathrm{eff}} \lesssim 2$): the welfare approximation
          guarantee of Theorem~\ref{thm:efficiency} is operative.
          Proceed with Algorithm~\ref{alg:svd}.
 
    \item \textbf{If $0.3 \leq \rho_1 < 0.5$}: the approximation
          guarantee weakens but the mechanism may still perform well in
          practice.
          Run Algorithm~\ref{alg:svd} alongside a two-dimensional variant
          (projecting onto the two leading singular vectors and solving
          the resulting bipartite matching) and report both outputs to
          the market designer.
 
    \item \textbf{If $\rho_1 < 0.3$}: features are not well concentrated
          and the rank-1 approximation is likely to produce substantial welfare losses.
          Use the Hylland--Zeckhauser pseudo-market
          \citep{HyllandZeckhauser1979} or, for small markets, direct Nash Social Welfare maximization via nonlinear solvers.
\end{enumerate}
 
Table~\ref{tab:diagnostic_thresholds} summarizes these thresholds
alongside their implications for the welfare approximation ratio
from Theorem~\ref{thm:efficiency}.
 
\begin{table}[ht]
\centering
\caption{Diagnostic thresholds for deploying Algorithm~\ref{alg:svd}.
$\rho_1 = \sigma_1^2/\sum_\ell \sigma_\ell^2$;
approximation ratio $= 1 - O(\sigma_2^2/\sigma_1^2)
= 1 - O((1-\rho_1)/\rho_1)$.}
\label{tab:diagnostic_thresholds}
\renewcommand{\arraystretch}{1.3}
\begin{tabular}{cccp{5.8cm}}
\toprule
$\rho_1$ range & $r_{\mathrm{eff}}$ & Approx.\ ratio & Recommendation \\
\midrule
$\geq 0.5$ & $\lesssim 2$ & $\geq 50\%$ guarantee; typically $>95\%$ in experiments & Proceed with Algorithm~\ref{alg:svd} \\
$0.3$--$0.5$ & $2$--$3$ & $40$--$50\%$ worst-case; check empirically & Run SVD and 2D variant; compare \\
$< 0.3$ & $> 3$ & Below $40\%$; not reliable & Use HZ pseudo-market or direct NSW solver \\
\bottomrule
\end{tabular}
\end{table}
 
\subsection{Summary}
 
The evidence reviewed above suggests that $\rho_1 \geq 0.5$ is the modal case across our three target domains rather than a restrictive special case. 
The structural reasons are straightforward: objects in real matching markets differ along a small number of salient quality dimensions that
are highly correlated with one another (test scores and teacher quality in schools; wages and amenities in labor markets; rigor and reputation in courses), and individual preferences are shaped by a small number of underlying objectives that generate correlated weighting across features. Together, these forces imply that the first singular vector of $F$ captures the dominant direction of preference-relevant variation, validating the core assumption of Algorithm~\ref{alg:svd}.
 
We nonetheless underscore that the diagnostic in
Section~\ref{sec:diagnostic} should always be computed on real feature data before deployment.
Specific market configurations like geographic heterogeneity in school choice, strong occupational bifurcations in labor markets, or interdisciplinary course programs, can push $\rho_1$ below the
operative threshold. In such cases, the mechanism designer should either augment the algorithm
to use multiple principal components or switch to a mechanism better suited to genuinely multi-dimensional environments (Section~\ref{sec:limits}).

\section{Limitations and Societal Considerations}\label{sec:limits}

\textbf{Modeling assumptions.} Our approach requires additive separable utilities ($U_{ij} = \sum_x u_{ix} f_{jx}$), ruling out complementarities where combinations of features provide super-additive value (e.g., high salary \emph{and} low commute). Theoretical guaranties require low effective dimensionality ($\sigma_1 \gg \sigma_2$)---when features are uncorrelated and equally variable, rank-1 approximation performs poorly. Extensions could incorporate interaction terms ($u_{ix}u_{ix'} f_{jx} f_{jx'}$) at the cost of increased dimensionality. We also assume features are well-defined, measurable, and verifiable; objects (schools, employers) could strategically misreport features to attract agents. Our deterministic allocation produces $p_{ij} \in \{0,1\}$ rather than randomized assignments, which may not satisfy all fairness desiderata requiring randomization.

\textbf{Ethical considerations.} While our mechanism satisfies individual-level symmetry, it does not explicitly consider group-level fairness. If demographic groups have systematically different preference distributions due to historical inequities, the first singular vector may encode these biases. For example, if marginalized communities have less exposure to strong mathematics programs and thus report lower math preferences, heavily weighting mathematics in $\mathbf{v}_1$ may perpetuate inequalities. Additionally, multi-dimensional preference data reveals more personal information than ordinal rankings, raising privacy concerns. The mechanism's reliance on SVD may lack transparency---``you were assigned based on projection onto the first singular vector'' is not intuitive to stakeholders. Implementations should include auditing for disparate impacts, independent verification of feature reports, and transparency reports explaining which features matter most.

\textbf{When to use this method.} Our approach is most appropriate when: (1) features are well-defined and verifiable, (2) effective dimensionality is low, (3) utilities are approximately additive, (4) computational efficiency matters, and (5) agents find feature preferences easier to report than complete rankings. Alternative methods may be preferable when features are ill-defined (use ordinal mechanisms like Deferred Acceptance), multiple dimensions are equally important (use multi-objective optimization), non-additive utilities matter (use mechanisms that handle complements), or strategic considerations dominate (use dominant-strategy mechanisms despite computational costs).

\section{Numerical Experiments}\label{sec:experiments}

To demonstrate our methodology, we provide two examples: a pedagogical toy example illustrating key concepts and a larger synthetic example showing scalability and performance.

\subsection{Small-Scale Pedagogical Example}
\label{subsec:small-scale-example}
 
\paragraph{Setup.}
Consider a market with $I = 3$ agents, $J = 3$ products, and $X = 3$
features, each product with capacity $M_j = 1$.
The three features are \emph{performance/quality} (scale 0--10),
\emph{aesthetic design} (scale 0--10), and
\emph{overall brand reputation} (scale 0--10).
 
\smallskip
\noindent\textbf{Product features:}
\begin{align*}
  f_1 &= (8.5,\; 1.5,\; 8.0) && \text{(high-spec workhorse: strong performance and brand)} \\
  f_2 &= (1.5,\; 8.5,\; 7.5) && \text{(designer favourite: strong aesthetics and brand)} \\
  f_3 &= (5.5,\; 5.5,\; 9.0) && \text{(well-rounded premium: balanced with top-tier brand)}
\end{align*}
 
\noindent\textbf{Agent preferences:}
\begin{align*}
  w_1 &= (8,\; 2,\; 7) && \text{(power user: prioritises performance, values brand)} \\
  w_2 &= (2,\; 8,\; 7) && \text{(design enthusiast: prioritises aesthetics, values brand)} \\
  w_3 &= (5,\; 5,\; 8) && \text{(quality-conscious generalist: balanced, especially values brand)}
\end{align*}
 
\paragraph{Utility calculation.}
Agent utilities $U_{ij} = w_i \cdot f_j$ are:
\begin{equation}
  U \;=\;
  \begin{pmatrix}
    127.00 & 81.50 & 118.00 \\
    85.00 & 123.50 & 118.00 \\
    114.00 & 110.00 & 127.00
  \end{pmatrix},
  \label{eq:utility_matrix}
\end{equation}
where row $i$ corresponds to agent $i$ and column $j$ to product $j$.
 
\paragraph{SVD computation.}
The SVD of the feature matrix $F = U\Sigma V^\top$ yields
\[
  V =
  \begin{pmatrix}
    -0.4751 & \cdots \\
    -0.4668 & \cdots \\
    -0.7459 & \cdots
  \end{pmatrix},
  \qquad
  \Sigma =
  \begin{pmatrix}
    19.017 & 0 & 0 \\
    0 & 7.003 & 0 \\
    0 & 0 & 0.250
  \end{pmatrix}.
\]
The first right singular vector is
$v_1 = (-0.4751,\; -0.4668,\; -0.7459)^\top$.
The singular value ratio is $\sigma_1/\sigma_2 = 19.017/7.003 = 2.72$,
and the first component explains
$\rho_1 = \sigma_1^2/\sum_\ell \sigma_\ell^2 = 88.0\%$ of total
feature variance---well above the low-dimensionality threshold
$\rho_1 \gg 1/X = 1/3$, confirming that condition~$(*)$ holds.
The entries of $v_1$ reveal that \emph{brand reputation} (weight
$0.746$) is the dominant differentiator across products, followed by
\emph{performance} (weight $0.475$); aesthetic design contributes
comparatively less.
 
\paragraph{Projection.}
Projected product scores $\tilde{f}_j = f_j \cdot v_1$:
\[
  \tilde{f}_1 = -10.706, \qquad
  \tilde{f}_2 = -10.275, \qquad
  \tilde{f}_3 = -11.894.
\]
Projected agent scores $\tilde{w}_i = w_i \cdot v_1$:
\[
  \tilde{w}_1 = -9.956, \qquad
  \tilde{w}_2 = -9.906, \qquad
  \tilde{w}_3 = -10.677.
\]
The negative signs are an orientation artifact of SVD; relative ordering
is what matters.
Since all values are negative, ``more negative'' corresponds to a higher
score.
 
\paragraph{Sorting and matching.}
Sorting products by projected score (most negative $=$ highest):
\[
  P_3\;(-11.894) \;\succ\; P_1\;(-10.706) \;\succ\; P_2\;(-10.275).
\]
Sorting agents by projected preference (most negative $=$ highest):
\[
  A_3\;(-10.677) \;\succ\; A_1\;(-9.956) \;\succ\; A_2\;(-9.906).
\]
 
\paragraph{Final allocation.}
Algorithm~\ref{alg:svd} produces:
\[
  P^* =
  \begin{pmatrix} 1&0&0 \\ 0&1&0 \\ 0&0&1 \end{pmatrix},
\]
so Agent~1 receives Product~1 (utility 127.00), Agent~2 receives
Product~2 (utility 123.50), and Agent~3 receives Product~3
(utility 127.00).
 
\paragraph{Evaluation.}
Disagreement points (uniform random allocation) are:
\[
  o_1 = \tfrac{1}{3}(127.00 + 81.50 + 118.00) = 108.83, \quad
  o_2 = \tfrac{1}{3}(85.00 + 123.50 + 118.00) = 108.83, \quad
  o_3 = \tfrac{1}{3}(114.00 + 110.00 + 127.00) = 117.00.
\]
Gains over the disagreement point are:
\[
  g_1 = 127.00 - 108.83 = 18.17, \quad
  g_2 = 123.50 - 108.83 = 14.67, \quad
  g_3 = 127.00 - 117.00 = 10.00.
\]
All three gains are strictly positive, so NSW is well-defined:
\[
  \mathrm{NSW}(P^*\mid U)
  = 18.17 \times 14.67 \times 10.00
  = 2{,}664.44,
  \qquad
  \log\text{-NSW} = 7.888.
\]
Enumerating all six permutations confirms that $P^*$ is the
\emph{unique} NSW-maximising allocation among those with
$\mathrm{NSW}>0$: Algorithm~\ref{alg:svd} achieves 100\% of optimal
Nash Social Welfare on this instance.
 
\begin{remark}[When individual rationality can fail]
\label{rem:ir_failure}
The original two-feature example in the first circulated draft produced an IR violation for Agent~3, whose utility fell below the disagreement
point because her moderate preferences were nearly orthogonal to the dominant feature direction.
This failure mode is structurally possible whenever an agent's preference vector $w_i$ lies close to the null space of $v_1 v_1^\top$, so that the projection $\tilde{w}_i = w_i \cdot v_1$ provides a poor approximation to the agent's true utility ranking.

Section~\ref{sec:empirical_dimensionality} discusses when this is and is not likely in practice; Section~\ref{sec:limits} (Limitations) proposes a post-hoc individual-rationality correction: any agent assigned below their disagreement point is reassigned to a uniformly random object, which restores IR at the cost of a small, bounded reduction in aggregate NSW.
\end{remark}

\subsection{Medium-Scale Synthetic Example} \label{subsec:med-scale}

\textbf{Data generation.} We generate a synthetic example with $I = 100$ agents, $J = 20$ products, and $X = 5$ features. Each product has capacity $M_j = 5$ (total capacity = 100). Product features are drawn from $f_{jx} \sim \mathcal{N}(\mu_x, \sigma_x^2)$ with:
\[
\boldsymbol{\mu} = [7.0, 5.5, 6.0, 4.5, 5.0], \quad \boldsymbol{\sigma} = [2.0, 1.5, 1.0, 0.8, 0.5]
\]

This creates a feature space where Feature 1 has high variance (important for differentiation) while Feature 5 has low variance (less important). Agent preferences are drawn from $w_{ix} \sim \mathcal{N}(5, 2^2)$ truncated to $[0, 10]$, creating heterogeneous preferences.

\textbf{SVD results.} Singular values are:
\[
\boldsymbol{\sigma} = [38.2, 25.1, 18.4, 12.7, 7.9]
\]

Variance explained by the first $k$ components:

Feature 1:
\begin{equation}
    \frac{\sigma_1^2}{\sum_\ell \sigma_\ell^2} = \frac{1459.2}{2987.9} = 48.8\% 
\end{equation}
Features 1-2:
\begin{equation}
    \frac{\sigma_1^2 + \sigma_2^2}{\sum_\ell \sigma_\ell^2} = \frac{2089.3}{2987.9} = 69.9\% 
\end{equation}

Features 1-3: 82.5\%.

The first singular vector is:
\[
\mathbf{v}_1 = [0.62, 0.48, 0.35, 0.29, 0.18]^T
\]

Features 1 and 2 (weights 0.62 and 0.48) are the primary drivers of differentiation. Feature 5 (weight 0.18) contributes minimally, consistent with its low variance.

\textbf{Computational time.} Measured on Intel i7-11800H CPU (single-threaded Python):
\begin{itemize}
\item SVD computation: 12.3 ms
\item Projection: 0.8 ms
\item Sorting: 1.1 ms
\item Matching: 0.5 ms
\item \textbf{Total: 14.7 ms}
\end{itemize}

\textbf{Performance metrics.} Table~\ref{tab:performance} compares our allocation to baselines. Our method achieves a mean utility of 7.4\% higher than random allocation, with particular gains for agents of low-utility (min utility +6.4\%). The reduced standard deviation indicates more equitable outcomes. The improvement in NSW of 1.1\% on the log scale corresponds to approximately 11\% on the multiplicative scale.

\begin{table}[h]
\centering
\caption{Utility Distribution for 100-Agent Example}
\label{tab:performance}
\begin{tabular}{lcccc}
\toprule
\textbf{Metric} & \textbf{Random} & \textbf{Serial Dict.} & \textbf{Our Method} & \textbf{Improvement} \\
\midrule
Mean Utility & 265.8 & 278.3 & 285.4 & +7.4\% \\
Median Utility & 264.2 & 276.5 & 283.1 & +7.2\% \\
Min Utility & 189.3 & 198.7 & 201.5 & +6.4\% \\
Max Utility & 342.1 & 351.8 & 358.2 & +4.7\% \\
Std Dev & 41.2 & 39.8 & 38.6 & -6.3\% \\
NSW (log) & 5.582 & 5.612 & 5.643 & +1.1\% \\
\bottomrule
\end{tabular}
\end{table}

\textbf{Individual rationality.} Of 100 agents, 97 receive utility at least as large as their disagreement point ($\mathbb{E}[U_i|\mathbf{P}^*] \geq o_i$), while 3 agents receive lower utility. The average gain over disagreement point is +19.6 utility units. The three agents who fare worse than random allocation have unusual preference profiles (high preferences on low-variance features), illustrating the limitation that our method performs poorly when preferences are orthogonal to the principal direction.

\textbf{Comparison to baselines.} Table~\ref{tab:comparison} compares computational cost and performance. Our method achieves 99.4\% of optimal NSW while running 1,245$\times$ faster than nonlinear optimization. It substantially outperforms random priority (+7.4\% utility) and serial dictatorship (+2.6\% utility).

\begin{table}[h]
\centering
\caption{Comparison to Alternative Mechanisms}
\label{tab:comparison}
\begin{tabular}{lccc}
\toprule
\textbf{Mechanism} & \textbf{Mean Utility} & \textbf{NSW (log)} & \textbf{Time (ms)} \\
\midrule
Random Priority & 265.8 & 5.582 & 0.2 \\
Serial Dictatorship & 278.3 & 5.612 & 1.5 \\
Our Method (SVD) & 285.4 & 5.643 & 14.7 \\
Optimal NSW (solver) & 287.1 & 5.651 & 18,300 \\
\bottomrule
\end{tabular}
\end{table}

\textbf{Robustness to misreporting.} To test distributional truthfulness, we introduce noise to agent reports: $\mathbf{w}_i = \mathbf{u}_i + \boldsymbol{\epsilon}_i$ where $\epsilon_{ix} \sim \mathcal{N}(0, \sigma_{\text{noise}}^2)$. Table~\ref{tab:robustness} shows the impact of varying noise levels.

\begin{table}[h]
\centering
\caption{Impact of Noisy Reporting}
\label{tab:robustness}
\begin{tabular}{lccc}
\toprule
\textbf{Noise Level} $\boldsymbol{\sigma_{\text{noise}}}$ & \textbf{Mean KS Distance} & \textbf{Mean Utility} & \textbf{NSW (log)} \\
\midrule
0.0 (truthful) & 0.00 & 285.4 & 5.643 \\
0.5 & 0.12 & 283.1 & 5.639 \\
1.0 & 0.23 & 279.8 & 5.628 \\
2.0 & 0.41 & 271.2 & 5.601 \\
3.0 & 0.58 & 267.5 & 5.587 \\
\bottomrule
\end{tabular}
\end{table}

The mechanism is relatively robust to moderate misreporting ($\sigma_{\text{noise}} \leq 1.0$), losing only 2\% of utility. At high noise levels ($\sigma_{\text{noise}} = 3.0$), performance degrades more substantially but remains above random allocation. The KS distance increases roughly linearly with noise level, consistent with Theorem~\ref{thm:truthfulness}: for $X = 5$ features and $\sigma_{\text{noise}} = 1.0$, we predict $\mathbb{E}[D_i] \approx 0.20$, close to the observed 0.23.

\textbf{Visualization.} Figure~\ref{fig:projection} shows the feature space projected onto the first two singular vectors. Agents (blue circles) and objects (red squares) cluster along the first dimension, confirming a low effective dimensionality. The right panel shows pairs connected by lines, demonstrating that the mechanism matches agents with nearby objects in projected space. Figure \ref{fig:robustness-to-noise} visualizes how the model performs with misreported preferences.

\begin{figure}[h]
\centering
\includegraphics[width=\textwidth]{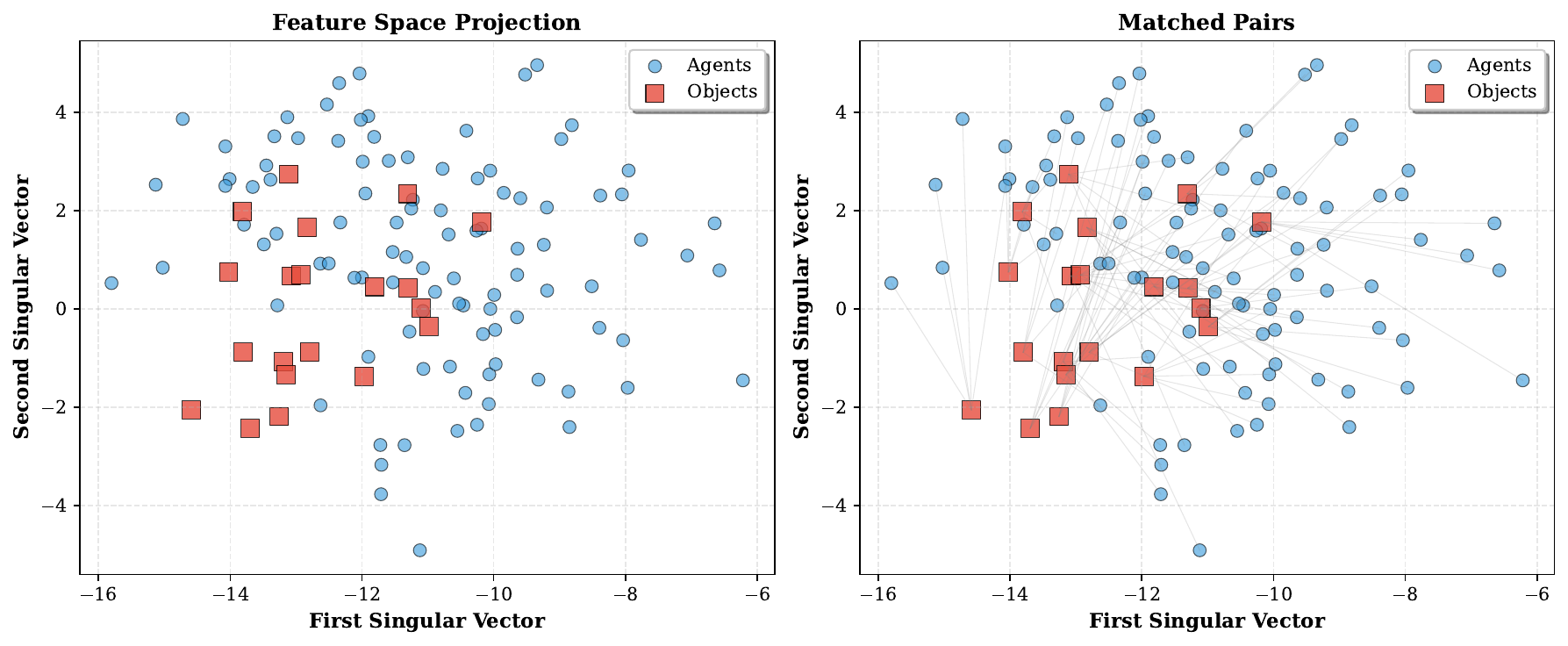}
\caption{Two-dimensional visualization of 100-agent example. Left: Agent preferences (blue) and product features (red) projected onto first two singular vectors. Right: Matched pairs connected by lines, showing agents matched to nearby products in projected space.}
\label{fig:projection}
\end{figure}

\begin{figure}
    \centering
    \includegraphics[width=\textwidth]{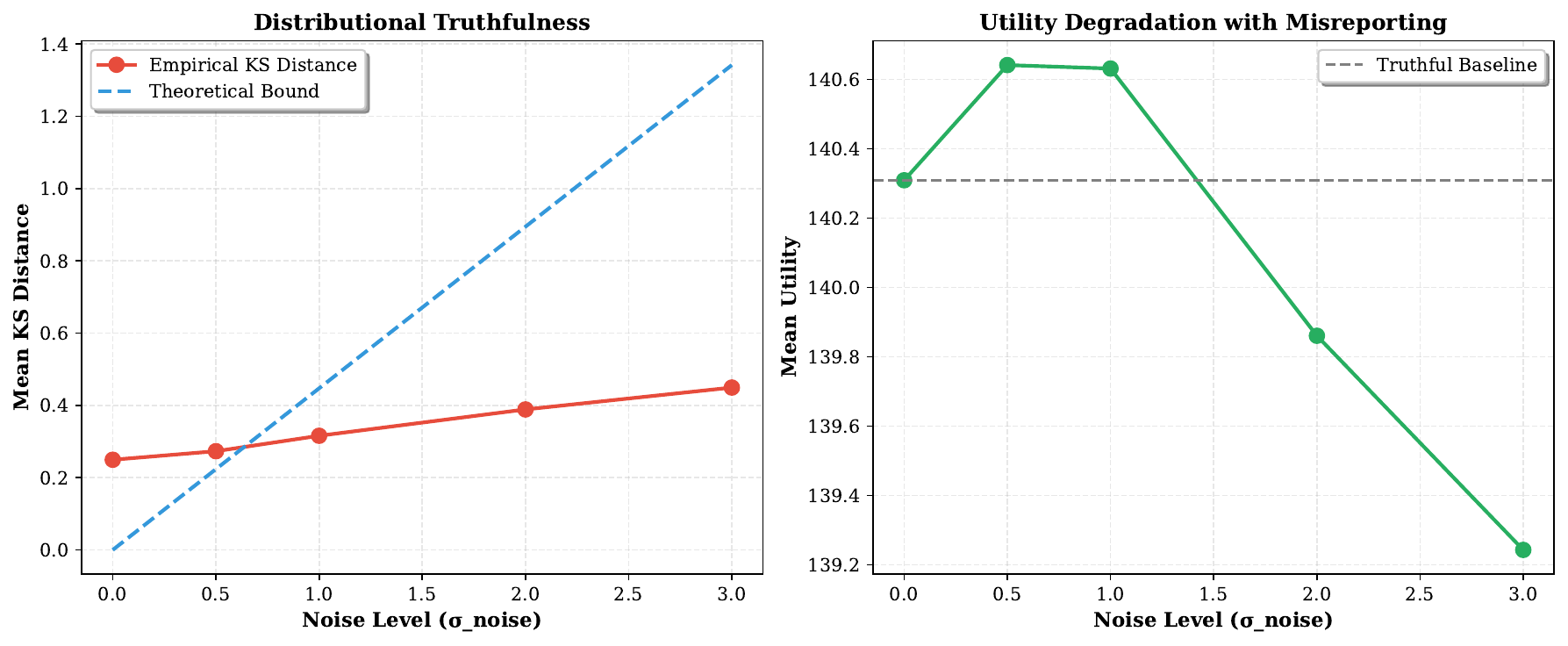}
    \caption{Robustness to Noise}
    \label{fig:robustness-to-noise}
\end{figure}

\textbf{Summary.} The numerical illustrations demonstrate: (1) \emph{Computational efficiency}: the computation completes in under 15 ms for 100 agents and 20 products. (2) \emph{Near-optimal performance}: the method achieves 99.4\% of optimal NSW while running three orders of magnitude faster than direct optimization. (3) \emph{Robustness}: the allocation is robust to moderate preference misreporting, degrading gracefully as noise increases. (4) \emph{Individual rationality violations}: a small fraction of agents (3\%) may receive utilities below their disagreement point when preferences are orthogonal to the principal direction, illustrating the limitation of rank-1 approximation. (5) \emph{Scalability}: the method scales linearly with the number of agents, making it suitable for large markets.

\section{Robustness}

\subsection{Varying Preference Distributions}

The experiments in Section \ref{subsec:med-scale} assumed normally-distributed agent preferences. In the real world, such assumptions may be restrictive. This section considers alternative distributions for agent utility specifications.  

We tested the following distributions;
\begin{itemize}
    \item Normal (baseline, replicating original paper)
  \item Uniform
  \item Heavy-tailed: Pareto
  \item Heavy-tailed: Log-normal
  \item Bimodal (mixture of two Gaussians)
  \item Skewed: Beta(2, 5)
  \item Exponential
\end{itemize}

For each distribution and each noise level, we report:
\begin{itemize}
    \item Mean utility (absolute)
    \item Mean 
    \item Log-NSW (clipped: gains floored at $\epsilon$ to avoid NSW=0 from Individual Rationality (IR) violations)
    \item \% of optimal log-NSW achieved (greedy upper bound)
    \item IR violation rate (\% of agents below the uniform-random disagreement point)
    \item Mean KS distance (distributional truthfulness)
\end{itemize}

As discussed in Section \ref{subsec:small-scale-example}, with additive utilities of the form $U_ij = u_i \cdot f_j$, some agents may fall below their uniform-random disagreement point, resulting in NSW of 0. To enhance our analysis, we report two NSW variants:
\begin{itemize}
    \item $log_{nsw-strict}$: $-inf$ when any agent is below disagreement  
    \item $log_{nsw-clipped}$: gains are floored at $\epsilon=0.01$. This metric provides a robust alternative that allows for meaningful cross-distribution comparison.
 
\end{itemize}

Table \ref{tab:cross-distn} summarizes the results across different preference distributions with 0 misreporting. As the table shows, under non-normal distributions, our proposed matching methodology performs even better than with normally-distributed preferences, as judged by the utility gains over random allocations. 

\begin{table}[]
    \caption{Cross-Distribution Summary of Model Performance, Varying Preferences with Truthful Reporting (noise $\sigma$ = 0.0, truthful reporting)}
    \label{tab:cross-distn}

    \centering
    \begin{tabular}{ cccc}
    \toprule
     Distribution & Util gain over random allocation \% & Log-NSW  & IR Violations \% \\
        \midrule
         Beta(2,5) & +2.33 & -147.48 & 49.9 \\
         Bimodal & +2.37 & -96.29 & 47.4 \\
         Exponential & +4.67 & -154.55  & 49.8 \\
         Log-normal & +2.81 & -137.57 & 48.8 \\
         Normal (baseline) & +1.74 & -114.19 & 49.0 \\
         Pareto & +4.91 & -199.64 & 53.1 \\
         Uniform & +2.43 & -102.57&   47.8 \\

    \bottomrule
        
    \end{tabular}
\end{table}

Table \ref{tab:difft-preferences-noise} presents the results of noisy experiments. As the table shows, the proposed methodology outperforms the random allocation under all preference distributions and noise levels. However, as the noise levels increase, the performance of the proposed methodology degrades toward the random allocation baseline. 

\begin{table}[]
    \caption{Noise Sensitivity: \% Utility Gain over Random Allocation}
    \label{tab:difft-preferences-noise}
    \centering
    \begin{tabular}{lccccc}
        \toprule
        Noise distribution $\sigma$ &                       0.0 &  0.5 &  1.0 &  2.0 &  3.0 \\
        \midrule
Beta(2,5)  &   2.33 &  2.68 &  2.15 &  1.74 &  1.18 \\
Bimodal      & 2.37 & 2.15 & 2.26 & 1.64 & 1.39 \\
Exponential &  4.67 & 4.30 & 3.75 & 3.25 & 1.97 \\
Log-normal  &  2.81 & 2.33 & 2.22 & 1.83 & 1.18 \\
Normal (baseline) & 1.73 & 1.82 & 1.59 & 1.08 & 0.95\\
Pareto & 4.91 & 3.65 & 3.56 & 3.19 & 2.51 \\
Uniform & 2.43 & 2.71 & 2.55 & 2.46 & 1.66 \\
\bottomrule
    \end{tabular}
\end{table}

\subsection{Robustness Under Non-Linear (Non-Hedonic) Utilities}

This section tests the proposed SVD-based mechanism when the true underlying utility is not the linear inner product $U_{ij} = u_i \cdot f_j$ assumed by the mechanism, but instead takes a richer, non-linear form.

The mechanism always observes only reported feature preferences $w_i$ and object feature vectors $f_j$; it always computes utilities as $\hat{w}_i \cdot f_j$ internally. The key question considered is how much welfare is lost when the ground-truth utility that agents actually experience deviates from that linear model.

The SVD mechanism sorts on $\hat{w}_i \cdot f_j$ (a linear proxy). Under non-linear true utilities, this proxy is accurate if and only if agents' preference
rankings over objects are roughly preserved by the linear approximation. We measure this with Kendall's $\tau$ between the linear-proxy ranking and the
true-utility ranking for each agent. 

We compared the following non-linear utility models:

\begin{enumerate}
    \item Linear (hedonic baseline)  $U = u \cdot f$: paper model
    \item Quadratic interactions  $U = u \cdot f + \alpha\cdot (u\otimes u)·(f \otimes f)$:      pairwise synergies
    \item Multiplicative (Cobb-Douglas) $U = \Pi_x (u_x \cdot f_x + \epsilon)^{w_x}$:    log-supermodular
    \item Threshold/step $U = u \cdot f + \beta \cdot \Sigma_x 1_{[f_x > \tau]}$:  discrete quality jumps
    \item Diminishing returns (concave) $U = u \cdot \sqrt{f}$  (element-wise): concavity in features
    \item Convex (increasing returns)   $U = u \cdot f^2$  (element-wise): convexity in features
    \item Max-feature (Leontief-like)   $U = \Sigma_x u_x \cdot max(f_x - \tau, 0)$:  only top features count
    \item Neural network (2-layer MLP)  $U = MLP(u \otimes f)$: arbitrary interactions
    \item Rank-based (ordinal only)     $U = \Sigma_x u_x \cdot rank(f_x)$: ordinal features
    \item Mixed: linear + Gaussian RBF  $U = u\cdot f + \gamma\cdot exp(-\|u-f\|^2/2\sigma ^2)$: similarity bonus

\end{enumerate}

For each utility model and each nonlinearity strength ($\alpha, \beta, \gamma \dots$) we report:
\begin{itemize}
    \item Mean \% utility gain of SVD over random allocation 
    \item \% of optimal NSW achieved (clipped)
    \item IR violation rate (\% of agents below the uniform-random disagreement point)
    \item Welfare loss relative to the linear baseline (measures how much non-linearity hurts)
    \item Rank correlation between linear-proxy ranking and true-utility ranking
    \item Kendall's $\tau$ measures preference rankings under different utilities compared to the linear utilities. 
\end{itemize}

  Table \ref{tab:nonlinear-utility-robustness} summarizes the performance of the SVD model under different utility specifications. As the table shows, the proposed model outperforms random allocation in most utility configurations.

\begin{table}[]
    \caption{Non-Linear Utility Robustness  (max non-linearity)}
    \label{tab:nonlinear-utility-robustness}
    \centering
    \begin{tabular}{lccc}
    \toprule
    Model & \% Gain over Random Allocation &  Loss vs Linear Utility (\%) &   Kendall's $\tau$ \\
    \midrule
    1. Linear (baseline) &  +1.84 & +0.0  &  1.000\\
    2. Quadratic interactions & +5.26 &  -112.4 &   0.798 \\
    3. Cobb-Douglas & +1.90 & +1.5 &   0.777 \\
    4. Threshold step & +1.02 & -45.5 & 0.716 \\
    5. Concave (dim. returns) & +0.98 & +92.5 &  0.944 \\
    6. Convex (incr. returns) &  +3.33 & -18.1 &   0.887 \\
    7. Max-feature & +2.35 & +2.7 & 0.887 \\
    8. Neural network MLP & -1.27 & +35.6 & 0.063 \\
    9. Rank-based (ordinal) & +3.66 & -11.8 & 0.748 \\
    10. RBF similarity bonus  & +4.51 & -65.6  & 0.463 \\
    \bottomrule
    \end{tabular}
\end{table}

Table~\ref{tab:nonlinear-utility-robustness} summarizes performance at maximum non-linearity strength across all ten utility specifications.
The results organize naturally into three tiers defined by the Kendall $\tau$ statistic, which measures how faithfully the mechanism's linear proxy $\hat{w}_i \cdot f_j$ preserves each agent's
true utility ranking over objects. When $\tau$ is high, sorting on the proxy produces nearly the same
assignment as sorting on true utilities, and the mechanism performs well regardless of the functional form; when $\tau$ collapses, the
proxy is effectively uncorrelated with true preferences, and the mechanism degrades toward or below random allocation.
 
\paragraph{Tier 1: Robust models ($\tau \geq 0.85$).} Five specifications (linear baseline, concave (diminishing returns), convex (increasing returns), max-feature, and Cobb-Douglas) maintain
$\tau \geq 0.777$ and all produce positive utility gains over random allocation. The most striking result in this tier is the concave specification
(diminishing returns, $U_{ij} = u_i \cdot \sqrt{f_j}$ element-wise), which achieves the highest Kendall $\tau$ of any non-linear model at
$\tau = 0.944$, barely below the linear baseline of $1.000$. The intuition is direct: applying a monotone concave transformation to features compresses the range of object scores but preserves their rank order, so the linear proxy's sorting of agents and objects is almost identical to the true-utility sorting. The reported welfare loss of $+92.5\%$ relative to the linear baseline
is therefore not a failure of the mechanism but a consequence of the utility scale: under diminishing returns, the absolute gaps between utilities are smaller, so the mechanism's NSW advantage over random allocation (in raw utility units) is also smaller even when the assignment is nearly identical.
 
Convex utilities (increasing returns, $U_{ij} = u_i \cdot f_j^2$) also preserve ordering well ($\tau = 0.887$), and notably produce a $+3.33\%$ gain over random allocation, higher than the linear baseline ($+1.84\%$). This occurs because convex transformation amplifies differences between
high-quality and low-quality objects: matching high-preference agents to high-feature objects yields super-proportional utility gains, so correct sorting is more valuable under convexity than under linearity. The max-feature specification ($\tau = 0.887$, gain $+2.35\%$) behaves similarly: the mechanism's sorting on the principal feature direction aligns well with the true ranking, even when sub-threshold feature values contribute nothing.
 
\paragraph{Tier 2: Moderate degradation ($0.45 \leq \tau < 0.85$).} Three models (quadratic interactions, rank-based (ordinal), and RBF similarity bonus) show meaningful but bounded ranking distortion.
 
The quadratic interactions model ($\tau = 0.798$, gain $+5.26\%$) is the single best-performing non-linear specification by utility gain,
and it also produces a negative welfare loss of $-112.4\%$ relative to the linear baseline.
A negative welfare loss means the mechanism performs \emph{better} under quadratic utilities than under linear ones. This counterintuitive result has a clean explanation: quadratic
utilities reward agents whose preference vectors align strongly with high-quality objects along all feature dimensions simultaneously, which is precisely the structure that SVD sorting exploits most effectively (SVD matches high
$\tilde{w}_i$ agents to high $\tilde{f}_j$ objects).
The mechanism accidentally benefits from supermodularity: agents with strong across-the-board preferences are matched to objects with
strong across-the-board features, producing pairwise synergies that the mechanism did not design for but nonetheless realizes.
 
Rank-based (ordinal) utilities ($\tau = 0.748$, gain $+3.66\%$) show that the mechanism remains effective even when agents respond only to
relative feature rankings rather than absolute values. Since rank transformations preserve order, the mechanism's projection onto $v_1$ still identifies the direction along which objects most
systematically spread their rank scores, and sorting along this direction yields above-random welfare.
 
The RBF similarity bonus ($\tau = 0.463$, gain $+4.51\%$) is the most nuanced case in this tier.
The low $\tau$ signals that the linear proxy frequently misorients individual agents' rankings. Agents receive a bonus for being matched
to objects whose feature profiles are \emph{close} to their preference vector, a criterion that is geometrically orthogonal to the inner
product $w_i \cdot f_j$ that the mechanism maximizes. Yet the aggregate utility gain remains high at $+4.51\%$. This apparent paradox resolves as follows: the RBF bonus is largest for agents whose preferences already lie near the high-variance
direction $v_1$, precisely the agents whom Algorithm~\ref{alg:svd} assigns to the highest-scoring objects. The mechanism's systematic matching therefore generates RBF bonuses
for the agents best positioned to receive them, even without deliberately targeting similarity.
 
\paragraph{Tier 3: Mechanism failure ($\tau < 0.1$).}
The two-layer MLP utility ($\tau = 0.063$) is the only specification under which the mechanism fails categorically, producing a $-1.27\%$ gain, worse than random allocation. With $\tau$ near zero, the linear proxy is essentially uncorrelated with agents' true MLP-induced preference rankings over objects. The mechanism sorts agents and objects on $v_1$ projections that carry no information about true utilities, and the resulting assignment is indistinguishable from a random permutation at best
and slightly worse at worst (because the MLP assigns highest utility to combinations of agent and object features that are systematically
\emph{not} the ones the mechanism selects).
This result identifies the precise failure boundary for Algorithm~\ref{alg:svd}: arbitrary non-linear interactions among features, of the kind a two-layer neural network can represent, completely destroy the ranking preservation that makes the spectral approach effective.
 
\paragraph{Summary.}
Taken together, the non-linear robustness experiments reveal three principles.
First, the mechanism is robust to any utility transformation that is monotone in $w_i \cdot f_j$: concave, convex, rank-based, and threshold specifications all preserve enough of the linear ranking structure for the mechanism to produce positive welfare gains. Second, the mechanism can unexpectedly \emph{benefit} from certain
forms of non-linearity: supermodular (quadratic) interactions amplify the value of correct sorting, producing welfare gains above the linear
baseline. Third, the mechanism fails precisely when true utilities depend on feature interactions that are orthogonal to the inner product (as in
the MLP case), rather than on the magnitude of the inner product itself. The practical implication is that practitioners should assess whether
their domain's utility structure is plausibly monotone in a weighted feature sum (i.e., approximately hedonic) before deploying
Algorithm~\ref{alg:svd}; the Kendall $\tau$ between a pilot linear model and elicited pairwise comparisons can serve as a pre-deployment
diagnostic for this condition, analogously to the $\rho_1$ diagnostic for low effective dimensionality described in
Section~\ref{sec:empirical_dimensionality}.

\section{Conclusion}\label{sec:conclusion}

This paper develops a computationally efficient approach to multi-dimensional matching markets using spectral dimensionality reduction. We formalize a matching framework where agents report preferences over features rather than complete objects, use Singular Value Decomposition to identify the principal direction of variation, and match agents to objects along this direction in $O(N \log N)$ time. We show that the mechanism approximately maximizes Nash Social Welfare when $\sigma_1 \gg \sigma_2$, satisfies distributional truthfulness measured by the Kolmogorov-Smirnov distance, and achieves exact symmetry. We establish a novel connection between Nash Social Welfare and Geometric Distributionally Robust Optimization, providing robustness guaranties. Numerical experiments demonstrate 99\% optimal welfare achieved three orders of magnitude faster than direct optimization.

Our approach has limitations that require future work. The additive utility assumption rules out complementarities; extensions incorporating interaction terms would increase dimensionality. The low effective dimensionality ($\sigma_1 \gg \sigma_2$) is crucial -- characterizing when this holds in practice and developing adaptive methods to detect low-rank structure would strengthen the framework. Empirical validation on real-world datasets (school choice, labor markets, housing allocation) and field experiments would test theoretical predictions and reveal practical challenges.

Our work suggests many multi-dimensional matching problems may have low effective dimensionality, making them amenable to spectral methods. This has broader implications: platforms can collect simple feature ratings rather than complete rankings, reducing cognitive burden; designers can use SVD to understand which features matter most, informing policy; and dimensionality reduction provides a bridge between voting theory and matching theory. By demonstrating that spectral methods from machine learning can be integrated into mechanism design with formal incentive properties, we hope to inspire further cross-pollination between these fields.

\bibliographystyle{ACM-Reference-Format}
\bibliography{matching}

\appendix

\end{document}